\documentclass{article}

\usepackage{fullpage}
\usepackage{amsmath,amsfonts,amssymb,amsthm,thmtools,enumerate}
\usepackage{hyperref,cleveref}

\usepackage[most]{tcolorbox}
\usepackage{xcolor}
\usepackage{orcidlink}

\newtcolorbox{aiusagebox}{
  colback=blue!3,
  colframe=blue!70!black,
  boxrule=1.4pt,
  arc=3mm,
  left=3mm,
  right=3mm,
  top=2mm,
  bottom=2mm
}

\providecommand{\Surj}{\mathsf{Surj}}
\providecommand{\Inj}{\mathsf{Inj}}
\providecommand{\Perm}{\mathsf{Perm}}
\providecommand{\Lin}{\mathsf{Lin}}
\providecommand{\Part}{\mathsf{Part}}
\providecommand{\Weak}{\mathsf{Weak}}
\providecommand{\Equiv}{\mathsf{Equiv}}
\providecommand{\ind}[1]{\mathbf1\{#1\}}
\DeclareMathOperator{\dom}{dom}
\DeclareMathOperator{\range}{range}
\DeclareMathOperator{\Maj}{Maj}
\DeclareMathOperator{\Sym}{Sym}
\DeclareMathOperator{\id}{id}
\DeclareMathOperator{\Neg}{neg}
\DeclareMathOperator*{\Ex}{\mathbb{E}}

\newtheorem{theorem}{Theorem}[section]
\newtheorem{lemma}[theorem]{Lemma}

\theoremstyle{definition}
\newtheorem{definition}[theorem]{Definition}
\theoremstyle{remark}

\title{Classification aggregation: a quantitative impossibility theorem}
\author{Yuval Filmus \orcidlink{0000-0002-1739-0872} \\
\small Taub Faculty of Computer Science and Faculty of Mathematics \\
\small Technion --- Israel Institute of Technology, Haifa, Israel \\
\small \texttt{yuvalfi@cs.technion.ac.il}}

\begin{document}

\maketitle
\begin{abstract}
A group of individuals wishes to classify $m$ objects into $n$ categories in such a way that no class is left empty, a condition known as \emph{surjectivity}. The opinions of the individuals are aggregated separately for each object using an aggregation function that can depend on the object.

Maniquet and Mongin showed that if the aggregation functions are unanimous and the outcome must always be surjective, then the aggregation mechanism is dictatorial. Cailloux et al.\ showed that the same holds even if unanimity is relaxed to citizen sovereignty (each object can be classified into any category).

We show that similar results hold even if we only require the outcome to be surjective with probability $1-\epsilon$ (with respect to an arbitrary symmetric i.i.d.\ distribution), provided that the aggregation functions are far from being constant.

On the way, we characterize all aggregation mechanisms whose outcome is always surjective without any assumptions on the aggregation functions.

Our approach uses a general result of Alekseev and Filmus which has wider applicability. We illustrate this by proving a similar impossibility result for aggregating equivalence relations.
\end{abstract}

\section{Introduction}
\label{sec:introduction}

Arrow's celebrated impossibility theorem~\cite{Arrow50,Arrow63} concerns elections among $n \ge 3$ candidates in which each voter has to rank all candidates, and the desired output is a ranking of all candidates. The theorem states that the only voting mechanisms satisfying the two natural axioms IIA and Unanimity are dictatorships.

One might hope that the issue goes away if we allow the outcome, encoded as a list of pairwise rankings of the candidates, to be non-transitive (that is, not correspond to a ranking) with small probability (under the impartial culture distribution). Kalai~\cite{Kalai} ruled this out for neutral mechanisms, showing that if the probability of a non-transitive outcome is small, then the voting mechanism is close to a dictatorship. Mossel~\cite{Mossel12} replaced the assumption of neutrality by the assumption that the mechanism is far from being constant. This assumption is necessary to rule out mechanisms which output a fixed ranking.

We prove a result in the style of Mossel in a different domain: classification aggregation. A group of individuals wishes to classify $m$ objects into $n$ categories in such a way that no category is left empty, a condition known as \emph{surjectivity}. The opinions of the individuals are aggregated into a societal outcome on an object-by-object basis, a constraint which mirrors the IIA axiom. This problem can be traced back to Kasher and Rubinstein~\cite{KasherRubinstein}, who considered the case $n = 2$.

Maniquet and Mongin~\cite{ManiquetMonginTR,ManiquetMongin} showed that if $m \ge n \ge 3$ and the mechanism satisfies Unanimity (if all individuals classify object $i$ into category $j$ then the outcome also classifies $i$ into category $j$) and Surjectivity (the outcome is always surjective), then the mechanism is a dictatorship, up to a permutation of the categories. Cailloux et al.~\cite{CHOS} proved a similar result under the weaker assumption of Citizen Sovereignty (for every object $i$ and category $j$, there is an outcome which classifies $i$ into category $j$), as well as for $m > n = 2$. When $m = n = 2$, other mechanisms exist, and Cailloux et al.\ described all of them.

Our main result shows that we cannot escape these impossibility theorems by relaxing Surjectivity.

\begin{theorem}[Main] \label{thm:main}
Let $m \ge n \ge 2$ such that $(m,n) \neq (2,2)$. Let $\mu$ be a full support distribution on the set of surjective mappings $[m] \to [n]$ which is invariant under applying permutations of $[n]$. We encode such mappings as vectors in $[n]^m$.

There exists $\epsilon_0 > 0$ (depending on $m,n,\mu$) for which the following holds.

For every $\epsilon > 0$ such that $\epsilon < \epsilon_0$ there exists $\delta > 0$ such that the following holds for all $p$. Suppose that $f_1,\dots,f_m\colon [n]^p \to [n]$ are \emph{$\delta$-surjective with respect to $\mu$}:
\[
 \Pr_{x^{(1)},\dots,x^{(p)} \sim \mu}[(f_1(x^{(1)}_1, \dots, x^{(p)}_1),\dots,f_m(x^{(1)}_m,\dots,x^{(p)}_m)) \text{ is surjective}] \ge 1 - \delta.
\]

If $f_1,\dots,f_m$ are \emph{$\epsilon$-far from constant} (for every $i \in [m]$ and $c \in [n]$, $\Pr[f_i=c] < 1-\epsilon$, where the probability is with respect to the uniform distribution) then there exists $k \in [p]$ and a permutation $\pi$ on $[n]$ such that for all $i \in [m]$,
\[
 \Pr_{x \sim [n]^p}[f_i(x) = \pi(x_k)] \ge 1-\epsilon,
\]
where the probability is with respect to the uniform distribution.
\end{theorem}

To explain the theorem, consider the following picture:
\[
\begin{array}{c|ccc}
x^{(1)} & x^{(1)}_1 & \cdots & x^{(1)}_m \\
\vdots & \vdots & \ddots & \vdots \\
x^{(p)} & x^{(p)}_1 & \cdots & x^{(p)}_m \\\hline
& f_1(x^*_1) & \cdots & f_m(x^*_m)
\end{array}
\]

Each row represents the classification of one of the individuals: for each of the $m$ objects, it includes a category in $[n]$. Each row has to be surjective, meaning that all categories must be present. The individual classifications are aggregated using the functions $f_1,\dots,f_m$, which operate on each column separately.

The result of Maniquet and Mongin states that if (1) the outcome is always surjective and (2) $f_i(j,\dots,j) = j$ for all $i,j$, then there exists $k \in [p]$ and a permutation $\pi$ on $[n]$ such that $f_1(x) = \cdots = f_m(x) = \pi(x_k)$. Cailloux et al.\ prove a similar impossibility result under (2'): for all $i,j$ there is $x$ such that $f_i(x) = j$.

\Cref{thm:main} considers the experiment in which each individual samples a classification from a distribution $\mu$ whose support is the set of all surjective classifications, and is invariant under permutations on $[n]$. It states that if (1') the outcome is \emph{almost} always surjective and (2'') the $f_i$ are not approximately constant, then $f_1,\dots,f_m$ are close to a function of the form $\pi(x_k)$. Condition (2'') is necessary to rule out functions such as $f_i(x) = \min(i,n)$.

A crucial feature of our theorem is that $\delta$ depends on $\epsilon$ and on $m,n,\mu$ but not on $p$.

\subsection{Outline of proof}
\label{sec:intro-outline}

The proof uses a general result of Alekseev and Filmus~\cite{AlekseevFilmus} on \emph{approximate polymorphisms}. Polymorphism is a concept from universal algebra whose relevance to judgment aggregation first appears explicitly in unpublished work of Szegedy and Xu~\cite{SzegedyXu15}. It already appears implicitly a few decades earlier in work of Rubinstein and Fishburn~\cite{RubinsteinFishburn}.

\begin{definition}[Polymorphism] \label{def:polymorphism}
Let $\Sigma$ be a finite alphabet. A \emph{predicate} is a subset $P \subseteq \Sigma^m$ for some $m$. A \emph{(multi-sorted) polymorphism of $P$} is a tuple of functions $f_1,\dots,f_m\colon \Sigma^p \to \Sigma$ (for some $p$) satisfying the following condition:
\[
x^{(1)},\dots,x^{(p)} \in P \Longrightarrow
(f_1(x^{(1)}_1, \dots, x^{(p)}_1),\dots,f_m(x^{(1)}_m,\dots,x^{(p)}_m)) \in P.
\]
\end{definition}

In classification aggregation the alphabet is $\Sigma = [n]$, and the predicate is $\Surj_{m,n}$:

\begin{definition}[Surjectivity]
\label{def:surj}
Let $m \ge n$. The predicate $\Surj_{m,n} \subseteq [n]^m$ consists of all vectors $(x_1,\dots,x_m) \in [n]^m$ such that $\{x_1,\dots,x_m\} = [n]$.
\end{definition}

Stated in this language, Maniquet and Mongin determined all polymorphisms of $\Surj_{m,n}$ for $m \ge n \ge 3$ which satisfy $f_i(j,\dots,j) = j$ for all $i,j$.

Historically, polymorphisms were studied mostly in the uni-sorted setting, which corresponds to the condition $f_1 = \cdots = f_m$, but the multi-sorted setting does appear in the literature~\cite{BulatovJeavons03}.

The result of Alekseev and Filmus concerns \emph{approximate polymorphisms}.

\begin{definition}[Approximate polymorphism]
\label{def:approximate-polymorphism}
Let $P \subseteq \Sigma^m$, and let $\mu$ be a distribution on $P$. A tuple of functions $f_1,\dots,f_m\colon \Sigma^p \to \Sigma$ is a \emph{$\delta$-approximate polymorphism (with respect to $\mu$)} if
\[
 \Pr_{x^{(1)},\dots,x^{(p)} \sim \mu}[(f_1(x^{(1)}_1, \dots, x^{(p)}_1),\dots,f_m(x^{(1)}_m,\dots,x^{(p)}_m)) \in P] \ge 1 - \delta.
\]
\end{definition}

The result of Alekseev and Filmus states that all approximate polymorphisms are close to exact polymorphisms. However, Alekseev and Filmus are only able to prove it for predicates satisfying an additional condition.

\begin{definition}[Flexibility]
\label{def:flexible}
A predicate $P \subseteq \Sigma^m$ is \emph{flexible} if for all $i \in [m]$ there is $x \in P$ such that $y \in P$ for all $y \in \Sigma^m$ which differ from $x$ only in the $i$'th coordinate.
\end{definition}

\begin{theorem}[Alekseev--Filmus]
\label{thm:alekseev-filmus}
Let $P \subseteq \Sigma^m$ be a flexible predicate, and let $\mu$ be a distribution on $P$ of full support. Denote by $\mu|_1,\dots,\mu|_m$ its marginals.

There exists $\epsilon_0 > 0$ (depending on $P,\mu$) for which the following holds.

For every $\epsilon > 0$ such that $\epsilon < \epsilon_0$ there exists $\delta > 0$ such that the following holds. If $f_1,\dots,f_m\colon \Sigma^p \to \Sigma$ are a $\delta$-approximate polymorphism of $P$ then there exists an exact polymorphism $g_1,\dots,g_m\colon \Sigma^p \to \Sigma$ of $P$ such that for all $i$,
\[
 \Pr_{x_1,\dots,x_p \sim \mu|_i}[f_i(x_1,\dots,x_p) \neq g_i(x_1,\dots,x_p)] < \epsilon.
\]

If $|\Sigma| = 2$, then the same holds even without the assumption of flexibility.
\end{theorem}

It is easy to check that $\Surj_{m,n}$ is flexible for all $m > n$, but not when $m = n$. We conjecture that the theorem holds without the assumption of flexibility for any $\Sigma$.

In order to prove \Cref{thm:main} when $m > n$, we apply \Cref{thm:alekseev-filmus} on $\Surj_{m,n}$ with the given distribution $\mu$. Since $\mu$ is invariant under applying permutations of $[n]$, each marginal $\mu|_i$ is the uniform distribution on $[n]$. When $m > n$, we can thus extend \Cref{thm:main} to arbitrary full support distributions $\mu$, with the caveat that the uniform distributions on $[n]^p$ need to be replaced by $p$-th powers of $\mu|_1,\dots,\mu|_m$.

The case $m = n$ requires a completely different argument, which is spectral in nature and uses ideas from work of Alon et al.~\cite{ADFS} on the traffic light problem. 

\medskip

In order to apply \Cref{thm:alekseev-filmus} we need to know all polymorphisms of $\Surj_{m,n}$, not just ones satisfying Unanimity or Citizen Sovereignty. This is an inherent feature of this setting: the assumption that $f_1,\dots,f_m$ are unanimous doesn't guarantee that they are close to a unanimous polymorphism. Indeed, take
\[
 f_i(x_1,\dots,x_p) =
 \begin{cases}
     j & \text{if } x_1 = \cdots = x_p = j, \\
     \min(i,n) & \text{otherwise}.
 \end{cases}
\]
The probability that $x^{(i)}_1 = \cdots = x^{(i)}_p$ for some $i$ is vanishingly small, and so $f_1,\dots,f_m$ are an $o(1)$-approximate polymorphism of $\Surj_{m,n}$, with respect to any $\mu$. However, they are not close to any dictatorship.

The same issue also arises in Mossel's quantitative Arrow theorem~\cite{Mossel12}, which is the special case of \Cref{thm:alekseev-filmus} for the predicate $\Lin_m \subseteq \{0,1\}^{\binom{m}{2}}$ of all linear orders. Mossel~\cite{Mossel09} resolves it by determining all polymorphisms of $\Lin_m$. Similarly, proving \Cref{thm:main} requires us to determine all polymorphisms of $\Surj_{m,n}$.

\subsection{Impossibility domains without unanimity}
\label{sec:intro-impossibility}

Dokow and Holzman~\cite{DokowHolzman09,DokowHolzman10binary} were interested in classifying predicates for which an impossibility theorem in the style of Arrow holds.

\begin{definition}[Impossibility domain]
\label{def:impossibility}
A predicate $P \subseteq \{0,1\}^m$ is an \emph{impossibility domain} if the only unanimous polymorphisms of $P$ are dictators. That is, if $f_1,\dots,f_m\colon \{0,1\}^p \to \{0,1\}$ constitute a polymorphism of $P$ satisfying $f_i(j,\dots,j) = j$ for all $i,j$, then there exists $k \in [p]$ such that $f_1(x) = \cdots = f_m(x) = x_k$.
\end{definition}

\begin{theorem}[Dokow--Holzman, as refined by Szegedy--Xu]
\label{thm:dokow-holzman}
If a predicate $P \subseteq \{0,1\}^m$ is not an impossibility domain then one of the following holds:
\begin{enumerate}[(i)]
\item $x_1 \oplus x_2 \oplus x_3, \dots, x_1 \oplus x_2 \oplus x_3$ is a (unanimous) polymorphism of $P$.
\item $\Maj(x_1,x_2,x_3), \dots, \Maj(x_1,x_2,x_3)$ is a (unanimous) polymorphism of $P$.
\item There is a unanimous polymorphism $f_1,\dots,f_m$ of $P$ in which $f_1,\dots,f_m \in \{x_1 \land x_2, x_1 \lor x_2\}$.
\end{enumerate}
\end{theorem}

Dokow and Holzman~\cite{DokowHolzman10} considered the case of larger alphabet under the stronger assumption of supportiveness ($f_i(x_1,\dots,x_p) \in \{x_1,\dots,x_p\}$), but were unable to provide a complete classification. Szegedy and Xu~\cite{SzegedyXu15} used tools from universal algebra to complete the classification, and also to prove a similar classification under unanimity.

\begin{theorem}[Szegedy--Xu]
\label{thm:szegedy-xu}
If a predicate $P \subseteq \Sigma^m$ is not an impossibility domain with respect to supportiveness, then this is witnessed by a polymorphism of arity $p=3$.

If a predicate $P \subseteq \Sigma^m$ is not an impossibility domain with respect to unanimity, then this is witnessed by a polymorphism of arity $p=|\Sigma|$.
\end{theorem}

Without unanimity, a new type of polymorphism arises: a \emph{certificate}. In the case of $\Surj_{m,n}$, a polymorphism $f_1,\dots,f_m$ is a certificate if for each $j \in [n]$, one of the $f_i$'s is the constant $j$ function.

\begin{definition}[Trivial domain]
\label{def:trivial}
A predicate $P \subseteq \Sigma^m$ is \emph{trivial} if for every polymorphism $f_1,\dots,f_m\colon \Sigma^p \to \Sigma$, one of the following cases holds:
\begin{enumerate}[(i)]
\item \textbf{$f_1,\dots,f_m$ are dictatorial:} There exists an index $k \in [p]$ and a permutation $\pi$ on $\Sigma$ such that
\[
 f_1(x) = \cdots = f_m(x) = \pi(x_k).
\]
\item \textbf{$f_1,\dots,f_m$ are a certificate:} There exists a partial mapping $\rho$ from $[m]$ to $\Sigma$ such that
\begin{enumerate}[(a)]
\item $\rho$ is a certificate: if $x_i = \rho(i)$ for all $i \in \dom\rho$ then $x \in P$.
\item $f_1,\dots,f_m$ conform to $\rho$: $f_i(x) = \rho(i)$ for all $i \in \dom\rho$ and all $x$.
\end{enumerate}
\end{enumerate}
We say that $P$ is \emph{$p$-trivial} if the condition holds for a specific value of $p$.
\end{definition}

Continuing our previous example, a certificate for $\Surj_{m,n}$ is a partial mapping $\rho$ from $[m]$ to $[n]$ which is surjective. Functions $f_1,\dots,f_m$ conform to $\rho$ if $f_i$ is the constant $\rho(i)$ function for any $i \in \dom\rho$.

Filmus~\cite{Filmus26} proves a result in the style of \Cref{thm:szegedy-xu} in this context. The result applies to predicates satisfying two mild technical conditions.

\begin{definition}[Degeneracy]
\label{def:degeneracy}
A predicate $P \subseteq \Sigma^m$ is \emph{non-degenerate} if
\begin{enumerate}[(a)]
\item $P$ has full projections: for every $i \in [m]$ and $\sigma \in \Sigma$ there exists $x \in P$ such that $x_i = \sigma$.
\item $P$ depends on all coordinates: for every $i \in [m]$ there are $x \in P$ and $y \notin P$ which differ only on the $i$'th coordinate.
\end{enumerate}
\end{definition}

\begin{definition}[Vulnerability]
\label{def:vulnerability}
A predicate $P \subseteq \Sigma^m$ is \emph{vulnerable} if there exists $i \in [m]$ and $\sigma \in \Sigma$ such that all $x \in \Sigma^m$ satisfying $x_i = \sigma$ belong to $P$.
\end{definition}

It is easy to check that $\Surj_{m,n}$ is non-degenerate and non-vulnerable.

In the setting of Filmus~\cite{Filmus26}, which allows different alphabets for different coordinates, the assumption of non-degeneracy is without loss of generality: we can shrink the alphabet of each coordinate to its projection, and drop coordinates which $P$ does not depend on. Vulnerability is needed to rule out non-trivial predicates such as OR (which has ORs as polymorphisms).

We can now state a consequence of the main result of Filmus~\cite{Filmus26}.

\begin{theorem}[Filmus]
\label{thm:filmus}
Let $P \subseteq \Sigma^m$ be a non-vulnerable non-degenerate predicate.

\begin{enumerate}[(a)]
\item If $|\Sigma| > 2$ and $P$ is $1$-trivial then it is trivial.
\item If $\Sigma = \{0,1\}$ and $P$ is $1$-trivial then it is trivial unless there is a polymorphism $f_1,\dots,f_m\colon \{0,1\}^2 \to \{0,1\}$ such that $f_1,\dots,f_m \in \{x_1 \land x_2, x_1 \lor x_2\}$.
\end{enumerate}
\end{theorem}

Using this result, we are able to determine all polymorphisms of $\Surj_{m,n}$, enabling us to deduce \Cref{thm:main} via \Cref{thm:alekseev-filmus}.

\begin{theorem}[Polymorphisms of $\Surj_{m,n}$] \label{thm:surj-polymorphisms}
Let $m \ge n \ge 2$, where $(m,n) \neq (2,2)$. If $f_1,\dots,f_m\colon [n]^p \to [n]$ constitute a polymorphism of $\Surj_{m,n}$ then one of the following holds:
\begin{enumerate}[(i)]
\item \label{itm:surj-dictator}
There exists $k \in [p]$ and a permutation $\pi$ of $[n]$ such that
\[
 f_1(x) = \cdots = f_m(x) = \pi(x_k).
\]
\item For each $j \in [n]$ there exists $i \in [m]$ such that $f_i(x) = j$ for all $x$.
\label{itm:surj-constant}
\end{enumerate}
\end{theorem}

The case $m = n$ was already proved by Falik and Friedgut~\cite{FalikFriedgut}. Using \Cref{thm:filmus}, we provide a quick proof of this case.

\subsection{Equivalence relations}
\label{sec:intro-equivalence}

Fishburn and Rubinstein~\cite{FishburnRubinstein} considered a different kind of classification problem, which also appears later in Kasher and Rubinstein~\cite{KasherRubinstein}. In this problem, a society of individuals is grouping a set of objects into equivalence classes. Each individual encodes their classification by stating, for any two objects, whether they are equivalent. The individual classifications are then aggregated by determining separately for any two objects whether they are equivalent.

In the language of polymorphisms, this corresponds to the following predicate.

\begin{definition}[Equivalence predicate]
\label{def:equivalence}
The predicate $\Equiv_m \subseteq \{0,1\}^{\binom{m}{2}}$ consists of all vectors $x$ satisfying the condition
\[
 x_{\{i,j\}} = x_{\{j,k\}} = 1 \Longrightarrow x_{\{i,k\}} = 1
\]
for all distinct $i,j,k$.
\end{definition}

We think of elements of $x \in \Equiv_m$ as equivalence relations in which $i \sim j$ if $x_{\{i,j\}} = 1$.

Fishburn and Rubinstein, and later Kasher and Rubinstein, showed that all unanimous polymorphisms of $\Equiv_m$ are oligarchies, that is, $f_1 = \cdots = f_m$, and the common value is a conjunction of inputs or a constant function.

We prove an analog of \Cref{thm:main} for this predicate, which we state using the language of polymorphisms.

\begin{theorem}[Approximate polymorphisms of $\Equiv_m$] \label{thm:main-equiv}
Let $m \ge 3$ and let $\mu$ be a full support distribution on $\Equiv_m$ which is invariant under permutations of $[m]$.

Let $q$ be the probability that $i \sim j$ under $\mu$, and let $\nu$ be the product distribution with marginals $q$. 

There exists $\epsilon_0 > 0$ (depending on $m$ and $\mu$) for which the following holds.

For every $\epsilon > 0$ such that $\epsilon < \epsilon_0$ there exists $\delta > 0$ such that the following holds for all $p$. Suppose that $(f_{\{i,j\}})_{i<j}\colon \{0,1\}^p \to \{0,1\}$ are a $\delta$-approximate polymorphism of $\Equiv_m$ with respect to $\mu$.

If $\Pr_\nu[f_{\{i,j\}}=0] < 1-\epsilon$ for all $i < j$ then there exists a set $S \subseteq [p]$ such that for all $i < j$,
\[
 \Pr_\nu\bigl[f_{\{i,j\}}(x) = \bigwedge_{s \in S} x_s\bigr] \ge 1 - \epsilon.
\]
\end{theorem}

This result follows from \Cref{thm:alekseev-filmus} via the following characterization of all polymorphisms of $\Equiv_m$.

\begin{theorem}[Exact polymorphisms of $\Equiv_m$]
\label{thm:equiv-polymorphisms}
Let $m \ge 3$. If $(f_{\{i,j\}})_{i<j}\colon \{0,1\}^p \to \{0,1\}$ constitute a polymorphism of $\Equiv_m$ then there exists an equivalence relation $\sim$ on $[m]$ such that the following holds:
\begin{enumerate}[(a)]
\item If $i \nsim j$ then $f_{\{i,j\}} = 0$.
\item For every equivalence class $A$ of $\sim$ of size at least~$3$ there exists a set $S \subseteq [p]$ such that for every $i\neq j \in A$, we have
\[
 f_{\{i,j\}}(x) = \bigwedge_{s \in S} x_s.
\]
\end{enumerate}
(If $A = \{i,j\}$ is an equivalence class of size~$2$ then $f_{\{i,j\}}$ can be arbitrary.)

Moreover, for every equivalence relation $\sim$ on $[m]$, all functions $f_{\{i,j\}}$ satisfying the above constraints are indeed polymorphisms.
\end{theorem}

The interpretation is that there is an equivalence relation $\sim$ such that the aggregated outcome is always a refinement of $\sim$. Moreover, if $A = \{i,j\}$ is an equivalence class of $\sim$ then $f_{\{i,j\}}$ is arbitrary, and if $|A| \ge 3$ then the restriction of $f$ to $A$ is an oligarchy which could depend on $A$.

The proof of \Cref{thm:equiv-polymorphisms} reduces to the case $m = 3$, which was proved by~\cite{Filmus26} as part of the classification of polymorphisms of all symmetric predicates over $\{0,1\}$ (when $m = 3$, the predicate $\Equiv_m$ consists of all vectors whose weight is not exactly~$2$).

\subsection{Related work}
\label{sec:related-work}

\Cref{thm:main} falls under the umbrella of \emph{approximate judgment aggregation}. Judgment aggregation~\cite{ListPettit,ListPuppe} arose from results such as the doctrinal paradox~\cite{KornhauserSager,ListPettit}, which in our language states that all unanimous polymorphisms of the predicate
\[
 P_\land = \{(a,b,a \land b) : a,b \in \{0,1\}\}
\]
are oligarchies, that is, functions of the form $f(x) = \bigwedge_{i \in S} x_i$.
The predicate $P_\land$ is \emph{truth-functional}, meaning that some coordinates are unconstrained, and the remaining ones are a function of the others.

There are several ways to relax the condition of polymorphicity. \Cref{thm:main} relaxes \emph{consistency}, which states that the result of applying the aggregation mechanism always satisfies the predicate.
There are several papers analyzing Arrow's theorem from this perspective~\cite{Kalai,Keller2010,Keller2012,Mossel12}.
Nehama~\cite{Nehama} proved an approximate version of the doctrinal paradox, but the bounds in his result depend on the number of voters. Filmus et al.~\cite{FLMM20} obtained an approximate doctrinal paradox in which the error parameters are independent of the number of voters. Chase et al.~\cite{CFMMS22} proved a similar result for all truth-functional predicates over $\{0,1\}$, and Alekseev and Filmus~\cite{AlekseevFilmus} generalized this to arbitrary predicates over $\{0,1\}$ and some predicates over larger alphabets.

Falik and Friedgut~\cite{FalikFriedgut} considered a different relaxation, in which the aggregation mechanism is only approximately coordinate-wise. They proved an approximate version of Arrow's theorem and of the Gibbard--Satterthwaite theorem~\cite{Gibbard,Satterthwaite} in this setting.

The Gibbard--Satterthwaite theorem assumes strategyproofness. Several papers prove approximate versions which relax this assumption~\cite{FKKN11,IKM12,MR15}.

\subsection*{Paper organization}
We derive \Cref{thm:filmus} from the results in~\cite{Filmus26} in \Cref{sec:triviality}. We prove \Cref{thm:main,thm:surj-polymorphisms} in \Cref{sec:classification}, and \Cref{thm:main-equiv,thm:equiv-polymorphisms} in \Cref{sec:equivalence}. We close the paper with a short discussion, \Cref{sec:discussion}.

\paragraph{Acknowledgments.} The author was supported by ISF grant 507/24.

\begin{aiusagebox}
\begin{center}
\textbf{AI usage statement}
\end{center}

\vspace{0.3em}

ChatGPT (GPT-5.5 Thinking, OpenAI, 30 April 2026) was used to assist in proving \Cref{lem:surj-polymorphisms-1}.
The author verified the correctness of the proof, and takes full responsibility
for the final content.
\end{aiusagebox}

\section{Triviality}
\label{sec:triviality}

In this section we derive \Cref{thm:filmus} from the main results of~\cite{Filmus26}. Our starting point is the following theorem, which results from a combination of~\cite[Theorem 1.3 and Theorem 1.4]{Filmus26} in the special case where all alphabets are the same.

\begin{definition}[$\Phi$-triviality]
\label{def:Phi-trivial}
Let $\Sigma$ be a finite alphabet, let $P \subseteq \Sigma^m$ be a predicate, and let $\Phi \subseteq \Sym(\Sigma)^m$, where $\Sym(\Sigma)$ is the group of permutations of $\Sigma$.

The predicate $P$ is \emph{$\Phi$-trivial} if for every polymorphism $f_1,\dots,f_m\colon \Sigma^p \to \Sigma$, one of the following cases holds:
\begin{enumerate}[(i)]
\item \textbf{$f_1,\dots,f_m$ are dictatorial:} There exist $k \in [p]$ and $(\phi_1,\dots,\phi_m) \in \Phi$ such that for all $i \in [m]$,
\[
 f_i(x) = \phi_i(x_k).
\]
\item \textbf{$f_1,\dots,f_m$ are a certificate}, as in \Cref{def:trivial}.
\end{enumerate}
We say that $P$ is \emph{$(p,\Phi)$-trivial} if this holds for a specific value of $p$.
\end{definition}

\begin{theorem}
\label{thm:filmus-paper}
Let $\Sigma$ be a finite alphabet, let $P \subseteq \Sigma^m$ be non-degenerate, and let $\Phi \subseteq \Sym(\Sigma)^m$.

If $P$ is $(1,\Phi)$-trivial but not $\Phi$-trivial, then one of the following holds:
\begin{enumerate}[(i)]
\item $P$ is \emph{sticky}: there exist $i \in [m]$ and $\sigma \in \Sigma$ such that if $x \in P$ then $x|_{i \to \sigma} \in P$, where $x|_{i \to \sigma}$ is obtained from $x$ by changing the $i$'th coordinate to $\sigma$. \label{itm:sticky}
\item $|\Sigma| = 2$ and there is a polymorphism $f_1,\ldots,f_m\colon \Sigma^2 \to \Sigma$ where all $f_i$ are unanimous ($f_i(\sigma,\sigma) = \sigma$ for all $\sigma \in \Sigma$) and non-dictatorial (cannot be written as a function of a single coordinate).

(When $\Sigma = \{0,1\}$, in this case $f_1,\ldots,f_m \in \{x_1 \land x_2, x_1 \lor x_2\}$.)
\label{itm:and-or}
\item $P$ has a polymorphism $f_1,\dots,f_m\colon \Sigma^2 \to \Sigma$ such that for all $x \in P$,
\[
 (f_1(x_1,\cdot),\dots,f_m(x_m,\cdot)),
 (f_1(\cdot,x_1),\dots,f_m(\cdot,x_m)) \in \Phi,
\]
where $f_i(a,\cdot)\colon \Sigma \to \Sigma$ is obtained by fixing the first coordinate to $a$, and $f_i(\cdot,a)$ is defined analogously.
\label{itm:latin-square}
\end{enumerate}
\end{theorem}

Comparing \Cref{thm:filmus-paper} and \Cref{thm:filmus}, we see that the stickiness obstruction is replaced by the weaker vulnerability obstruction, and that the last case of \Cref{thm:filmus-paper} is missing from \Cref{thm:filmus}. The proof of \Cref{thm:filmus} amounts to showing that stickiness implies vulnerability, and that ``Latin square polymorphisms'' cannot exist when $\phi_1 = \cdots = \phi_m$ for all $(\phi_1,\dots,\phi_m) \in \Phi$. We formalize this slightly more general result as a theorem, and then deduce \Cref{thm:filmus}.

\begin{theorem}
\label{thm:filmus-identical}
Let $\Sigma$ be a finite alphabet, let $P \subseteq \Sigma^m$ be non-degenerate, and let $\Phi \subseteq \{(\phi,\dots,\phi) : \phi \in \Sym(\Sigma)\}$.

If $P$ is $(1,\Phi)$-trivial but not $\Phi$-trivial, then one of the following holds:
\begin{enumerate}[(i)]
\item $P$ is vulnerable.
\item $|\Sigma| = 2$ and there is a polymorphism $f_1,\dots,f_m\colon \Sigma^2 \to \Sigma$ where all $f_i$ are unanimous and non-dictatorial.
\end{enumerate}
\end{theorem}

The proof uses a lemma.

\begin{lemma}
\label{lem:filmus-identical}
Let $\Sigma$ be a finite alphabet, let $P \subseteq \Sigma^m$ be non-degenerate, and let $\Phi \subseteq \{(\phi,\dots,\phi) : \phi \in \Sym(\Sigma)\}$.

If $P$ is $(1,\Phi)$-trivial and $|P| = |\Sigma|$ then $|\Sigma| = 2$ and there is a polymorphism $f_1,\dots,f_m\colon \Sigma^2 \to \Sigma$ such that all $f_i$ are unanimous and non-dictatorial.
\end{lemma}
\begin{proof}
By non-degeneracy, there are permutations $\pi_1,\dots,\pi_m \in \Sym(\Sigma)$ such that
\[
 P = \{(\pi_1(\sigma),\dots,\pi_m(\sigma)) : \sigma \in \Sigma\}.
\]

By non-degeneracy, $|\Sigma| \ge 2$. Suppose first that $|\Sigma| = 2$, without loss of generality $\Sigma = \{0,1\}$. Then $\pi_1,\dots,\pi_m \in \{\id,\Neg\}$, where $\id(x) = x$ and $\Neg(x) = \lnot x$. Define $f_1,\dots,f_m\colon \{0,1\}^2 \to \{0,1\}$ as follows:
\[
 f_i(x) =
 \begin{cases}
 x_1 \land x_2 & \text{if } \pi_i = \id, \\
 x_1 \lor x_2 & \text{if } \pi_i = \Neg.
 \end{cases}
\]

We claim that $f_1,\dots,f_m$ constitute a polymorphism of $P$. To check this, it suffices to consider $x^{(1)} = (\pi_1(0), \dots, \pi_m(0))$ and $x^{(2)} = (\pi_1(1), \dots, \pi_m(1))$. Then
\begin{itemize}
\item If $\pi_i = \id$ then $f_i(x^{(1)}_i,x^{(2)}_i) = 0 \land 1 = 0 = x^{(1)}_i$.
\item If $\pi_i = \Neg$ then $f_i(x^{(1)}_i,x^{(2)}_i) = 1 \lor 0 = 1 = x^{(1)}_i$.
\end{itemize}
Thus $(f_1(x^{(1)}_1,x^{(2)}_1), \dots, f_m(x^{(1)}_m,x^{(2)}_m)) = x^{(1)}$, showing that $f_1,\dots,f_m$ are indeed a polymorphism.

Suppose next that $|\Sigma| \ge 3$. Let $\sigma,\tau \in \Sigma$ be two different letters. Define $f_1,\dots,f_m\colon \Sigma^1 \to \Sigma$ as follows:
\[
 f_i(a) =
 \begin{cases}
     \pi_i(\sigma) & \text{if } a = \pi_i(\sigma), \\
     \pi_i(\tau) & \text{otherwise}.
 \end{cases}
\]
For every $b \in \Sigma$,
\[
 (f_1(\pi_1(b)), \dots, f_m(\pi_m(b))) =
 \begin{cases}
     (\pi_1(\sigma),\dots,\pi_m(\sigma)) & \text{if } b = \sigma, \\
     (\pi_1(\tau),\dots,\pi_m(\tau)) & \text{otherwise}.
 \end{cases}
\]
Therefore $f_1,\dots,f_m$ constitute a polymorphism of $P$. Since $|\Sigma| \ge 3$, each $f_i$ is neither constant nor a permutation, contradicting the $(1,\Phi)$-triviality of $P$.
\end{proof}

We can now complete the proof of \Cref{thm:filmus-identical}.
\begin{proof}[Proof of \Cref{thm:filmus-identical}]
In view of \Cref{thm:filmus-paper}, it suffices to show that stickiness implies vulnerability, and to rule out \Cref{itm:latin-square} of that \namecref{thm:filmus-paper}.

Suppose that $P$ is sticky: there exist $i \in [m]$ and $\sigma \in \Sigma$ such that $x \in P$ implies $x|_{i \to \sigma} \in P$. Consider the functions $f_1,\dots,f_m\colon \Sigma^1 \to \Sigma$ given by
\[
 f_j(\tau) =
 \begin{cases}
     \sigma & \text{if } j = i, \\
     \tau & \text{if } j \neq i.
 \end{cases}
\]
By assumption, $f_1,\dots,f_m$ constitute a polymorphism of $P$. Since $P$ is $(1,\Phi)$-trivial and $f_i$ is constant, $f_1,\dots,f_m$ must be a certificate. Since $f_i$ is the only constant function among $f_1,\dots,f_m$, it follows that the certificate only involves the $i$'th coordinate. Therefore any $x \in \Sigma^m$ satisfying $x_i = \sigma$ belongs to $P$, and so $P$ is vulnerable.

Next, we rule out \Cref{itm:latin-square}.
In view of \Cref{lem:filmus-identical}, we can assume that $|P| > |\Sigma|$ (otherwise \Cref{itm:and-or} holds), and so we can find $\sigma \in \Sigma$ such that there are at least two $x \in P$ with $x_1 = \sigma$. This implies that we can find $x,y \in P$ and an index $i > 1$ such that $x_1 = y_1$ and $x_i \neq y_i$.

Suppose that $f_1,\dots,f_m$ constitute a polymorphism of $P$ conforming to \Cref{itm:latin-square}. Since
\[
 (f_1(x_1,\cdot),\dots,f_m(x_m,\cdot)),
 (f_1(y_1,\cdot),\dots,f_m(y_m,\cdot)) \in \Phi
\]
we have, in particular,
\[
 f_i(x_i,\cdot) = f_1(x_1,\cdot) = f_1(y_1,\cdot) = f_i(y_i,\cdot).
\]

On the other hand,
\[
 (f_1(\cdot,x_1),\dots,f_m(\cdot,x_m)) \in \Phi,
\]
which implies that $f_i(\cdot,x_i)$ is a permutation. This contradicts $f_i(x_i,\cdot) = f_i(y_i,\cdot)$, which implies that $f_i(x_i,x_i) = f_i(y_i,x_i)$.
\end{proof}

\Cref{thm:filmus} follows immediately from \Cref{thm:filmus-identical} by taking
\[
 \Phi = \{(\phi,\dots,\phi) : \phi \in \Sym(\Sigma)\}.
\]

\section{Classification aggregation}
\label{sec:classification}

\subsection{Exact polymorphisms}
\label{sec:classification-exact}

In this section we prove \Cref{thm:surj-polymorphisms}. The proof rests on the following \namecref{lem:surj-polymorphisms-1}:

\begin{lemma} \label{lem:surj-polymorphisms-1}
\Cref{thm:surj-polymorphisms} holds when $p = 1$ and $n \ge 3$.
\end{lemma}

\Cref{thm:surj-polymorphisms} follows almost directly.

\begin{proof}[Proof of \Cref{thm:surj-polymorphisms}]
The case $n = 2$ (and so $m > 2$) follows directly from Filmus~\cite[Theorem 1.5]{Filmus26}.

The case $n \ge 3$ follows from \Cref{lem:surj-polymorphisms-1} by \Cref{thm:filmus}, once we verify that $\Surj_{m,n}$ is non-degenerate and non-vulnerable:
\begin{itemize}
\item $\Surj_{m,n}$ has full projections: this is easy to check.
\item $\Surj_{m,n}$ depends on all coordinates: it suffices to consider $x \in \Surj_{m,n}$ in which $x_i$ is the only coordinate equal to some $j$.
\item $\Surj_{m,n}$ is not vulnerable: this is because for any $j \in [n]$, the constant $j$ vector is not in $\Surj_{m,n}$. \qedhere
\end{itemize}
\end{proof}

In the rest of this section, we prove \Cref{lem:surj-polymorphisms-1}.

The case $m = n$ was proved by Falik and Friedgut~\cite{FalikFriedgut}. We provide a simple alternative proof.

\begin{proof}[Proof of \Cref{lem:surj-polymorphisms-1} when $m = n$]
Consider the $n \times n$ matrix $M$ defined by $M(i,j) = f_i(j)$. In other words, the $i$'th row of $M$ is the truth table of $f_i$.

For each permutation $x_1,\dots,x_n \in [n]^n$, we must have that $f_1(x_1),\dots,f_n(x_n)$ is also a permutation of $[n]$. This means that the ``generalized diagonal'' $M(1,x_1),\dots,M(n,x_n)$ is a permutation of $[n]$.
The permutations $i \mapsto i + t \pmod{n}$ cover each entry of $M$ exactly once, and so each of the numbers $1,\dots,n$ appears exactly $n$ times in $M$.

For any $i \in [n]$, any two entries of $M$ labeled $i$ must lie on the same row or on the same column, since otherwise we can construct a generalized diagonal passing through both of them.
Therefore either all $i$-entries are on the same row, or all $i$-entries are on the same column. Since there are exactly $n$ entries labeled $i$, this means that there is either a row whose entries are all labeled $i$, or a column whose entries are all labeled $i$.

Since every row and column intersect, it follows that either all rows are constant (and so all columns are equal to some permutation of $[n]$), or all columns are constant (and so all rows are equal to some permutation of $[n]$).

If all rows are constant then each $f_i$ is constant, and so $f_1,\dots,f_n$ are of certificate type.

If all columns are constant then $f_1,\dots,f_n$ are all equal to the same permutation, and so $f_1,\dots,f_n$ are dictatorial.
\end{proof}

The case $m > n$ is a lot more challenging, and will occupy the rest of this section.

We start with some definitions:
\begin{align*}
I &= \{ i \in [m] : f_i \text{ is constant} \}, \\
C &= \{ j \in [n] : f_i \equiv j \text{ for some } i \in I \}.
\end{align*}

Clearly,
\[
 |I| \ge |C|.
\]

We denote $c = |C|$ for short.

If $c = n$ then $f_1,\dots,f_m$ are of certificate type, so from now on we assume
\[
 c \leq n - 1.
\]

The entire proof is driven by the following \namecref{lem:hall}.

\begin{lemma}
\label{lem:hall}
For each $r \notin C$ there is a subset $X_r \subseteq [n]$ of size
\[
 c + 1 \le |X_r| \leq n - 1
\]
such that the set
\[
 I_r := \{ i \in [m] : f_i(s) = r \text{ for all } s \in X_r \}
\]
satisfies $I_r \cap I = \emptyset$ and
\[
 |I_r| \ge m - |X_r| + 1.
\]
\end{lemma}

Before we prove the \namecref{lem:hall}, let us explain what it means. It gives a witness showing that if $x \in \Surj_{m,n}$, then necessarily $f_1(x_1),\dots,f_m(x_m)$ contains $r$. Indeed, there are at least $|X_r|$ coordinates of $x$ whose value is in $X_r$. Since $|I_r| \ge m - |X_r| + 1$, one of these coordinates belongs to $I_r$. If $x_i \in X_r$ and $i \in I_r$ then by definition, $f_i(x_i) = r$.
We show that such a witness must always exist using Hall's theorem.

\begin{proof}
Consider the following bipartite graph:
\begin{itemize}
\item The left side is $[n]$.
\item The right side is $[m]$.
\item We connect $s \in [n]$ to $i \in [m]$ if $f_i(s) \neq r$.
\end{itemize}

We claim that the graph does not have a matching saturating the left side. Indeed, suppose that $a\colon [n] \to [m]$ were such a matching. We construct $x \in \Surj_{m,n}$ as follows:
\begin{itemize}
\item If $i \in \range a$, say $i = a(s)$, then $x_i = s$. (This guarantees that $x \in \Surj_{m,n}$.)
\item If $i \notin \range a$ then we choose $x_i$ so that $f_i(x_i) \neq r$. This is possible since otherwise $r \in C$.
\end{itemize}
By construction, $f_i(x_i) \neq r$ for all $i \in [m]$, contradicting the assumption that $f_1,\dots,f_m$ constitute a polymorphism.

Since there is no matching saturating the left side, by Hall's theorem there is a subset $X_r \subseteq [n]$ such that $|N(X_r)| \leq |X_r| - 1$, where $N(X_r)$ is the neighborhood of $X_r$.

We claim that $I_r = [m] \setminus N(X_r)$. Indeed, $i \in N(X_r)$ iff $f_i(s) \neq r$ for some $s \in X_r$, and so $i \notin N(X_r)$ iff $f_i(s) = r$ for all $s \in X_r$. It follows that
\[
 |I_r| = m - |N(X_r)| \ge m - |X_r| + 1.
\]

Next, we show that $I_r \cap I = \emptyset$. Since $|N(X_r)| < |X_r|$, the set $X_r$ is non-empty, say $s \in X_r$. If $i \in I_r$ then $f_i(s) = r$. Since $r \notin C$, this implies that $i \notin I$.

Thus $|I_r| \leq m - |I| \leq m - c$, which implies that
\[
 |X_r| \ge m - |I_r| + 1 \ge c + 1.
\]

Finally, we show that $|X_r| \leq n - 1$. Indeed, otherwise we have $X_r = [n]$, and so $I_r$ must be empty (since $i \in I_r$ would imply $f_i \equiv r$). However,
\[
 |I_r| \ge m - |X_r| + 1 = m - n + 1 \geq 1. \qedhere
\]
\end{proof}

Observe that since $c \leq n - 1$, applying the \namecref{lem:hall} to any $r \notin C$ implies the stronger inequality
\[
 c \leq n - 2.
\]

The next two steps give upper and lower bounds on the quantity
\[
 \Delta := \sum_{r \notin C} |I_r| |X_r|.
\]

We will obtain the lower bound via \Cref{lem:hall}.
The upper bound will be a consequence of the elementary observation that for each $i \in [m]$, the sets $X_r$ such that $i \in I_r$ are disjoint.
We will show that the two bounds are consistent only if $c = 0$, in which case they coincide. This implies that both of them are tight. The tightness of the lower bound turns out to imply that for each $r$, the set $X_r$ is a singleton $\{x_r\}$ and so $I_r = [m]$, implying that $f_i(x_r) = r$ for all $i \in [m]$. This will complete the proof. 

We start with the upper bound.

\begin{lemma}
\label{lem:Delta-ub}
We have
\[
 \Delta \leq (m-c)n.
\]
\end{lemma}
\begin{proof}
We claim that for each $i \in [m]$, the sets $X_r$ such that $i \in I_r$ are disjoint. Indeed, suppose that $i \in I_r \cap I_s$ and $t \in X_r \cap X_s$, where $r \neq s$. Then $f_i(t) = r$ (since $i \in I_r$ and $t \in X_r$) and $f_i(t) = s$ (since $i \in I_s$ and $t \in X_s$), and we reach a contradiction.

We deduce that
\[
 \Delta = \sum_{r \notin C} \sum_{i \in I_r} |X_r| = \sum_{i \notin I} \sum_{r\colon i \in I_r} |X_r| \leq |[m] \setminus I| \, n,
\]
using $I_r \cap I = \emptyset$. The \namecref{lem:Delta-ub} now follows from $|I| \geq |C| = c$.
\end{proof}

The lower bound uses discrete concavity.

\begin{lemma}
\label{lem:Delta-lb}
We have
\[
 \Delta \ge (n-c) \cdot \min\bigl((c+1)(m-c), (n-1)(m-n+2)\bigr)
\]
and
\[
 \Delta \ge (n-c)m.
\]

Furthermore, if the latter bound is tight then $c = 0$ and for all $r \in [n]$, we have $|X_r| = 1$ and $I_r = [m]$.
\end{lemma}
\begin{proof}
\Cref{lem:hall} implies that for each $r \notin C$,
\[
 |I_r| |X_r| \ge \theta(|X_r|), \text{ where } \theta(x) = x(m-x+1).
\]
Moreover, $c + 1 \leq |X_r| \leq n - 1$.

The function $\theta(x)$ is strictly concave, and so $\theta(|X_r|) \geq \min(\theta(c+1), \theta(n-1))$. Together with $|[n] \setminus C| = n-c$, this gives the first inequality.
The second inequality follows similarly using $1 \leq |X_r| \leq m$.

By strict concavity, if the second inequality is tight then $|X_r| \in \{1,m\}$ for each $r \notin C$. Since $|X_r| \leq n - 1 < m$, we see that $|X_r| = 1$ for all $r \notin C$. Since $|X_r| \ge c + 1$, this implies that $c = 0$ and so $C = \emptyset$. Finally, \Cref{lem:hall} shows that $|I_r| \ge m$, and so $I_r = [m]$, for all $r \in [n]$.
\end{proof}

The next step is the most technical: we rule out $c \ge 1$ by combining \Cref{lem:Delta-ub} with the first lower bound in \Cref{lem:Delta-lb}. This requires two technical inequalities. In both cases, we will use the inequality
\[
 2(n - 1) > n,
\]
which follows from $n > 2$.

\begin{lemma}
\label{lem:c=0-aux1}
If $c \ge 1$ then
\[
 (n-c) \cdot (c+1)(m-c) > (m-c) n.
\]
\end{lemma}
\begin{proof}
Since $m-c \ge m-n+2 > 0$, it suffices to show that $(c+1)(n-c) > n$.

The left-hand side is concave in $c$, and so $1 \leq c \leq n-2$ implies
\[
 (c+1)(n-c) \ge 2(n-1) > n. \qedhere
\]
\end{proof}

\begin{lemma}
\label{lem:c=0-aux2}
If $c \ge 1$ then
\[
 (n-c) \cdot (n-1)(m-n+2) > (m-c) n.
\]
\end{lemma}
\begin{proof}
If we subtract the right-hand side from the left-hand side then we obtain
\[
 x + c \bigl(n - (n-1)(m-n+2)\bigr),
\]
where $x$ doesn't involve $c$. Since
\[
 (n-1)(m-n+2) \ge 2(n-1) > n,
\]
we see that the validity of the \namecref{lem:c=0-aux2} for all $c \leq n-2$ would follow from its validity for $c = n-2$.

When $c = n-2$, the inequality in the statement of the \namecref{lem:c=0-aux2} reads
\[
 2(n-1)(m-n+2) > (m-n+2)n,
\]
which follows from $2(n-1) > n$ since $m-n+2 > 0$.
\end{proof}

We can now show that $c = 0$.

\begin{lemma}
\label{lem:c=0}
We have $c = 0$.
\end{lemma}
\begin{proof}
Suppose, for the sake of contradiction, that $c \ge 1$. Combining \Cref{lem:Delta-ub} and the first inequality in \Cref{lem:Delta-lb} gives
\[
 (m-c)n \ge \Delta \ge (n-c) \cdot \min\bigl((c+1)(m-c), (n-1)(m-n+2)\bigr).
\]
This is ruled out by \Cref{lem:c=0-aux1,lem:c=0-aux2}.
\end{proof}

We can now complete the proof of \Cref{lem:surj-polymorphisms-1}. \Cref{lem:Delta-ub} and the second inequality in \Cref{lem:Delta-lb} give
\[
 mn \ge \Delta \ge nm,
\]
and so $\Delta = nm$. \Cref{lem:Delta-lb} thus implies that $|X_r| = 1$ and $I_r = [m]$ for all $r \in [n]$, say $X_r = \{x_r\}$.

By definition of $I_r$, this shows that $f_i(x_r) = r$ for all $i \in [m]$. Clearly $x_r \neq x_s$ for $r \neq s$, and so this completely defines $f_i$ for all $i \in [m]$. We see that all $f_i$ are equal, and the common value is the permutation which is the inverse of $r \mapsto x_r$. Thus $f_1,\dots,f_m$ are dictatorial.

\subsection{Approximate polymorphisms (\texorpdfstring{$m > n$}{m > n})}
\label{sec:classification-approximate-flexible}

In this section we deduce \Cref{thm:main} for $m > n$ by combining \Cref{thm:surj-polymorphisms} with \Cref{thm:alekseev-filmus}.

\begin{proof}[Proof of \Cref{thm:main} when $m > n$]
We claim that since $m > n$, the predicate $\Surj_{m,n}$ is flexible, and so we can apply \Cref{thm:alekseev-filmus}. Indeed, for every $i \in [m]$ we can find $x \in \Surj_{m,n}$ such that $\{ x_j : j \neq i \} = [n]$, so that $x \in \Surj_{m,n}$ regardless of $x_i$.

Let $\epsilon_0 > 0$ be the parameter given by \Cref{thm:alekseev-filmus} for $\Surj_{m,n}$ and $\mu$.
Given $\epsilon > 0$ such that $\epsilon < \epsilon_0$, we let $\delta > 0$ be the parameter given by \Cref{thm:alekseev-filmus}. Since $\mu$ is invariant under applying permutations of $[n]$, the distributions $\mu|_1,\dots,\mu|_m$ are uniform.

Suppose that $f_1,\dots,f_m\colon [n]^p \to [n]$ are $\delta$-surjective with respect to $\mu$, in other words are a $\delta$-approximate polymorphism of $\Surj_{m,n}$ with respect to $\mu$, and moreover are $\epsilon$-far from constant.

\Cref{thm:alekseev-filmus} implies that there is a polymorphism $g_1,\dots,g_m\colon [n]^p \to [n]$ of $\Surj_{m,n}$ such that for all $i$,
\[
 \Pr[g_i = f_i] \ge 1-\epsilon.
\]

By \Cref{thm:surj-polymorphisms}, $g_1,\dots,g_m$ are either dictatorial or a certificate. If they are dictatorial, we are done, so it remains to show that they cannot be a certificate.

Indeed, if $g_1,\dots,g_m$ were a certificate then for some $i \in [m]$, the function $g_i$ would be the constant~$1$ function. However, in that case $f_i$ would be $\epsilon$-close to constant, contradicting the assumption that $f_1,\dots,f_m$ are $\epsilon$-far from constant.
\end{proof}

\subsection{Approximate polymorphisms (\texorpdfstring{$m = n$}{m = n})}
\label{sec:classification-approximate-perm}

In this section we complete the proof of \Cref{thm:main} by settling the case $m = n \ge 3$. In this case, we write $\Perm_n$ for $\Surj_{n,n}$; this is the set of all permutations of $[n]$, viewed as elements in $[n]^n$. Moreover, $\mu$ is the uniform distribution over $\Perm_n$, since this is the only distribution over $\Perm_n$ which is invariant under applying permutations of $[n]$.

The argument relies on the work of Alon et al.~\cite{ADFS}, who use Hoffman's bound to bound the size of independent sets in $K_n^p$, and to characterize independent sets of maximal and almost maximal size.

In order to make the argument more transparent, we start by sketching another proof of \Cref{thm:surj-polymorphisms}, which doesn't rely on \Cref{thm:filmus}. We then use the same approach to prove \Cref{thm:main}.

\subsubsection{Another proof of \texorpdfstring{\Cref{thm:surj-polymorphisms}}{Theorem \ref{thm:surj-polymorphisms}}}
\label{sec:Perm-poly}

For $m \le n$, let $\Inj_{m,n} \subseteq [n]^m$ (the Injectivity predicate) consist of all vectors $(x_1,\dots,x_m) \in [n]^m$ where $x_i \neq x_j$ for $i \neq j$. We can view $\Perm_n$ as $\Inj_{n,n}$. Since injectivity is a pairwise condition, in some sense it suffices to understand the case $m = 2$.

Let $n \ge 3$ and $f\colon [n]^p \to [n]$. Greenwell and Lov\'asz~\cite{GreenwellLovasz} showed that if $f,f$ constitute a polymorphism of $\Inj_{2,n}$ then $f$ is a dictator. We describe one of their proofs.

For any $r \in [n]$, we can consider the function $f_r\colon [n]^p \to \{0,1\}$ given by $f_r(x) = \ind{f(x) = r}$. This function satisfies the following condition: if $x,y \in [n]^p$ are such that $(x_i,y_i) \in \Inj_{2,n}$ for all $i \in [p]$, then $(f_r(x),f_r(y)) \neq (1,1)$. Equivalently, $f_r$ is the indicator function of an independent set in $K_n^p$. We say that $f_r$ is \emph{independent}.

Let $\mu(f_r)$ be the expectation of $f_r$ with respect to the uniform distribution. Greenwell and Lov\'asz show that if $f_r$ is independent then $\mu(f_r) \leq 1/n$, with equality iff $f_r(x) = \ind{x_k = s}$ for some $k \in [p]$ and $s \in [n]$. This easily implies that if $f,f$ constitute a polymorphism of $\Inj_{2,n}$ then $f$ is a dictator.

Alon et al.~\cite{ADFS} gave a different proof of the inequality $\mu(f_r) \leq 1/n$, using a technique known as Hoffman's bound~\cite{Hoffman,Friedgut,Haemers}. Their proof shows moreover that if $\mu(f_r) \ge 1/n - \epsilon$ then $f_r$ is $O(\epsilon)$-close to a function of the form $\ind{x_k = s}$, that is, $\Pr[f_r(x) \neq \ind{x_k = s}] = O(\epsilon)$. The latter bound is shown using an adaptation of the FKN theorem~\cite{FKN}, a classical result in Boolean function analysis.

In contrast to the setting studied by Greenwell and Lov\'asz and by Alon et al., in our case we have $n$ different functions $g_1,\dots,g_n$ (we use $g$ rather than $f$ to avoid a clash with the notation $f_r$). The argument of Alon et al.\ can be adapted to yield the following result.

\begin{lemma}
\label{lem:adfs-exact}
Let $n \ge 3$.

Suppose that $g_{i,r},g_{j,r}\colon [n]^p \to \{0, 1\}$ are \emph{cross-independent}: if $x,y \in [n]^p$ are such that $(x_i,y_i) \in \Inj_{2,n}$ for all $i \in [p]$, then $(g_{i,r}(x),g_{j,r}(y)) \neq (1,1)$. Then
\[
 (n - 1)^2 \mu(g_{i,r}) \mu(g_{j,r}) \leq (1 - \mu(g_{i,r})) (1 - \mu(g_{j,r})),
\]
with equality iff $g_{i,r}(x) = g_{j,r}(x) = \ind{x_k = s}$ for some $k \in [p]$ and $s \in [n]$.
\end{lemma}

We prove a stronger result, \Cref{lem:adfs-approx}, in \Cref{sec:Perm-approx-poly}.

In the rest of this subsection, we show how to derive \Cref{thm:surj-polymorphisms} for $m = n$ using \Cref{lem:adfs-exact}. The proof of \Cref{thm:main} for $m = n$ follows the same approach.

\smallskip

Let $g_1,\dots,g_n\colon [n]^p \to [n]$ be a polymorphism of $\Perm_n$ for $n \ge 3$. For $i,r \in [n]$, we define
\[
 g_{i,r}(x) = \ind{g_i(x) = r}, \quad \mu_{i,r} = \mu(g_{i,r}).
\]

If we put the numbers $\mu_{i,r}$ in an $n \times n$ matrix, then we obtain a bistochastic matrix $M$. Indeed, by definition $\mu_{i,r} \ge 0$; for all $i \in [n]$ we have
\begin{equation*} 
 \sum_{r \in [n]} \mu_{i,r} = 1
\end{equation*}
by definition; and for all $r \in [n]$ we have
\begin{equation} \label{eq:i-sum}
 \sum_{i \in [n]} \mu_{i,r} = 1,
\end{equation}
since the left-hand side is the expected number of $r$'s in $g_1(x^{(1)}_1,\dots,x^{(p)}_1), \dots, g_n(x^{(1)}_n,\dots,x^{(p)}_n)$, where the vectors $x^{(1)},\dots,x^{(p)} \in \Perm_n$ are drawn uniformly at random.

\Cref{lem:adfs-exact} gives an additional constraint: if $i \neq j$ then
\[
 (n - 1)^2 \mu_{i,r} \mu_{j,r} \leq (1 - \mu_{i,r}) (1 - \mu_{j,r}),
\]
and so
\begin{equation} \label{eq:prod-ub}
 n (n - 2) \mu_{i,r} \mu_{j,r} \leq 1 - \mu_{i,r} - \mu_{j,r}.
\end{equation}
Using these constraints, we are able to determine all possible values of the matrix $M$, and deduce the structure $g_1,\dots,g_n$.

Fix $i$ and $r$. Summing \Cref{eq:prod-ub} over all $j \neq i$ and using \Cref{eq:i-sum}, we obtain
\[
 n(n - 2) \mu_{i,r} (1 - \mu_{i,r}) \leq (n - 1)(1 - \mu_{i,r}) - (1 - \mu_{i,r}) = (n - 2)(1 - \mu_{i,r}).
\]
Therefore
\[
 \mu_{i,r} = 1 \quad \text{or} \quad \mu_{i,r} \leq \frac{1}{n}.
\]

Now fix $r$ and vary $i$. In view of \Cref{eq:i-sum}, we see that one of the following cases holds:
\begin{enumerate}[(i)]
\item For some $i_r$ we have $\mu_{i_r,r} = 1$ and $\mu_{i,r} = 0$ for $i \neq i_r$.

In this case,
\[
 g_{i_r} = r.
\]

\item For all $i$ we have $\mu_{i,r} = 1/n$.

The equality clause of \Cref{lem:adfs-exact} implies that in this case there are $k_r \in [p]$ and $s_r \in [n]$ such that for all $i$,
\[
 g_i(x) = r \text{ iff } x_{k_r} = s_r.
\]
\end{enumerate}

If the second case happens for any $r$ then $g_1,\dots,g_n$ are all non-constant. Considering all values of $r$, we deduce that one of the following cases holds:
\begin{enumerate}[(i)]
\item For every $r$ there exists $i_r$ such that $g_{i_r} = r$. This is \Cref{itm:surj-constant} of \Cref{thm:surj-polymorphisms}.

\item For every $r$ there exist $k_r,s_r$ such that for all $i$,
\[
 g_i(x) = r \text{ iff } x_{k_r} = s_r.
\]
This results in a contradiction unless all $k_r$ are equal and all $s_r$ are different. This is \Cref{itm:surj-dictator} of \Cref{thm:surj-polymorphisms}.
\end{enumerate}

\subsubsection{Proof of \texorpdfstring{\Cref{thm:main}}{Theorem \ref{thm:main}} when \texorpdfstring{$m = n$}{m = n}}
\label{sec:Perm-approx-poly}

We actually prove a stronger result: \Cref{thm:alekseev-filmus} holds for the predicate $\Perm_n$ when $n \ge 3$ and $\mu$ is the uniform distribution.

\begin{theorem}
\label{thm:filmus-Perm}
Let $n \ge 3$. There exists $\epsilon_0 > 0$ (depending on $n$) such that the following holds.

If $f_1,\dots,f_n\colon [n]^p \to [n]$ are an $\epsilon$-approximate polymorphism of $\Perm_n$ (with respect to the uniform distribution), where $\epsilon < \epsilon_0$, then one of the following cases holds:
\begin{enumerate}[(i)]
\item There exists $k \in [p]$ and a permutation $\pi$ of $[n]$ such that for all $i \in [n]$,
\label{itm:filmus-Perm-dict}
\[
 \Pr_{x \sim [n]^p}[f_i(x) \neq \pi(x_k)] = O(\sqrt[4]{\epsilon}).
\]
\item For each $j \in [n]$ there exists $i \in [n]$ such that
\label{itm:filmus-Perm-const}
\[
 \Pr_{x \sim [n]^p}[f_i(x) \neq j] = O(\sqrt[4]{\epsilon}).
\]
\end{enumerate}
In both cases, the probability is taken over the uniform distribution.
\end{theorem}

As in \Cref{sec:classification-approximate-flexible}, this implies \Cref{thm:main} with $\delta = \Theta(\epsilon^4)$.

In the remainder of this subsection, we prove \Cref{thm:filmus-Perm}.
Throughout the proof, we assume that $\epsilon$ is ``small enough'' by setting $\epsilon_0$ accordingly.

\smallskip

As in \Cref{sec:Perm-poly}, define
\[
 f_{i,r}(x) = \ind{f_i(x) = r}, \quad \mu_{i,r} = \mu(f_{i,r}).
\]
For $i \neq j$, the functions $f_{i,r},f_{j,r}$ are $\epsilon$-cross-independent.

\begin{definition}[Approximate cross-independence]
\label{def:approx-cross-independent}
The functions $f_{i,r},f_{j,r}\colon [n]^p \to \{0,1\}$ are \emph{$\epsilon$-cross-independent} if
\[
 \Pr_{x^{(1)},\dots,x^{(p)} \sim \Inj_{2,n}}[f_{i,r}(x^{(1)}_i,\dots,x^{(p)}_i) = f_{j,r}(x^{(1)}_j,\dots,x^{(p)}_j) = 1] \leq \epsilon,
\]
where $x^{(1)},\dots,x^{(p)}$ are drawn uniformly at random, and we index them using $i,j$.
\end{definition}

\begin{lemma}
\label{lem:approx-cross-independent}
If $f_1,\dots,f_n\colon [n]^p \to [n]$ are an $\epsilon$-approximate polymorphism of $\Perm_n$ then for every $i \neq j$ and $r$, the functions $f_{i,r},f_{j,r}$ are $\epsilon$-cross-independent.
\end{lemma}
\begin{proof}
Suppose that $x^{(1)},\dots,x^{(p)} \sim \Perm_n$. If we remove all indices except for $i,j$, then we get samples from $\Inj_{2,n}$ instead.

If $f_{i,r}(x^{(1)}_i,\dots,x^{(p)}_i) = f_{j,r}(x^{(1)}_j,\dots,x^{(p)}_j) = 1$ then $f_i(x^{(1)}_i,\dots,x^{(p)}_i) = f_j(x^{(1)}_j,\dots,x^{(p)}_j) = r$, and so this happens with probability at most $\epsilon$.
\end{proof}

Similarly, we have the following analog of \Cref{eq:i-sum}.

\begin{lemma} \label{lem:i-sum}
For every $r \in [n]$ we have
\[
 \sum_{i \in [n]} \mu_{i,r} = 1 \pm O(\epsilon).
\]
\end{lemma}
\begin{proof}
Define
\[
 N(x^{(1)},\dots,x^{(p)}) = \sum_{i \in [n]} f_{i,r}(x^{(1)}_i,\dots,x^{(p)}_i),
\]
which is the number of $r$'s in the ``output''. Observe that
\[
 \Ex_{x^{(1)},\dots,x^{(p)} \sim \Perm_n}[N(x^{(1)},\dots,x^{(p)})] = \sum_{i \in [n]} \mu_{i,r}.
\]

By definition, $N$ always lies in the range $\{0,\dots,n\}$. Since $f_1,\dots,f_n$ are an $\epsilon$-approximate polymorphism of $\Perm_n$, we have $N = 1$ with probability at least $1-\epsilon$. The \namecref{lem:i-sum} immediately follows.
\end{proof}

The main technical step of the proof is a generalization of \Cref{lem:adfs-exact}. The proof, which is essentially a version of Hoffman's bound, uses Fourier analysis on the group $\mathbb{Z}_n^r$.

\begin{lemma}
\label{lem:adfs-approx}
Let $n \ge 3$. Suppose that $f_{i,r},f_{j,r}$ are $\epsilon$-cross-independent. Then
\begin{enumerate}[(a)]
\item We have
\label{itm:adfs-ub}
\[
 n(n-2) \mu_{i,r} \mu_{j,r} \leq 1 - \mu_{i,r} - \mu_{j,r} + O(\sqrt{\epsilon}).
\]
\item If $\mu_{i,r},\mu_{j,r} \ge 1/n - \gamma$, where $0 \leq \gamma \leq 1/n$, then there exist $k_{i,r} = k_{j,r}$ and $s_{i,r} \neq s_{j,r}$ such that
\[
 \Pr_{x \sim [n]^p}[f_{i,r}(x) \neq \ind{x_{k_{i,r}} = s_{i,r}}] = O(\epsilon + \gamma),
\]
and similarly for $j$.
\label{itm:adfs-stab}
\end{enumerate}
\end{lemma}
\begin{proof}
\textbf{Setup.}
Every function on $[n]^p$ can be decomposed into its ``degree $d$'' parts:
\[
 f_{i,r} = \sum_{d=0}^p f_{i,r}^{=d},
\]
and similarly for $f_{j,r}$; see \cite[Chapter 8]{ODonnell}.

Each part $f_{i,r}^{=d}$ belongs to the space $V^{=d}$ of ``pure degree $d$ functions''. These spaces are orthogonal with respect to the inner product
\[
 \langle g,h \rangle = \Ex_{x \sim [n]^p}[g(x) h(x)],
\]
where the expectation is with respect to the uniform distribution. This implies that
\[
 \langle f_{i,r}, f_{j,r} \rangle = \sum_{d=0}^p \langle f_{i,r}^{=d}, f_{j,r}^{=d} \rangle,
\]
as well as
\[
 \|f_{i,r}\|^2 = \sum_{d=0}^p \|f_{i,r}^{=d}\|^2,
\]
where $\|g\| = \sqrt{\langle g,g \rangle}$. Note that
\[
 \|f_{i,r}\|^2 = \Ex_{x \sim [n]^p}[f_{i,r}(x)^2] = \mu_{i,r},
\]
since $f_{i,r}$ is Boolean (that is, $\{0,1\}$-valued).

The function $f_{i,r}^{=0}$ is the constant function equal to $\mu(f_{i,r})$. Defining $f_{i,r}^{>d} = \sum_{e>d} f_{i,r}^{=e}$, this implies that
\[
 \|f_{i,r}^{>0}\|^2 = \|f_{i,r}\|^2 - \|f_{i,r}^{=0}\|^2 = \mu_{i,r} - \mu_{i,r}^2 = \mu_{i,r} (1 - \mu_{i,r}).
\]

\smallskip
\textbf{Formula from \cite{ADFS}.} Let $x^{(1)},\dots,x^{(p)} \sim \Inj_{2,n}$. Alon et al.~\cite[Claim 4.3]{ADFS} proved the following formula:
\[
\Ex[f_{i,r}(x^{(1)}_i,\dots,x^{(p)}_i) f_{j,r}(x^{(1)}_j,\dots,x^{(p)}_j)] = \sum_{d=0}^p \left(-\frac{1}{n-1}\right)^d \langle f_{i,r}^{=d}, f_{j,r}^{=d} \rangle.
\]
Since $f_{i,r}$ and $f_{j,r}$ are Boolean, the left-hand side is the probability that $f_{i,r} = f_{j,r} = 1$. Therefore the assumption on $f_{i,r},f_{j,r}$ implies
\begin{equation} \label{eq:adfs}
\epsilon \ge \mu_{i,r} \mu_{j,r} + \sum_{d=1}^p \left(-\frac{1}{n-1}\right)^d \langle f_{i,r}^{=d}, f_{j,r}^{=d} \rangle.
\end{equation}

\smallskip
\textbf{Proof of \Cref{itm:adfs-ub}.}
If $\mu_{i,r} \leq \sqrt{\epsilon}$ then the inequality trivially holds, since $1 - \mu_{j,r} \geq 0$. Therefore we can assume that $\mu_{i,r},\mu_{j,r} \ge \sqrt{\epsilon}$.

Since $1/(n-1) \leq 1$, \Cref{eq:adfs} implies that
\[
 \epsilon \ge \mu_{i,r} \mu_{j,r} - \frac{1}{n-1} \sum_{d=1}^p \langle f_{i,r}^{=d}, f_{j,r}^{=d}\rangle = \mu_{i,r} \mu_{j,r} - \frac{1}{n-1} \langle f_{i,r}^{>0}, f_{j,r}^{>0} \rangle.
\]

Cauchy--Schwarz implies that
\[
 \epsilon \ge \mu_{i,r} \mu_{j,r} - \frac{1}{n-1} \|f_{i,r}^{>0}\| \|f_{j,r}^{>0}\| = 
 \mu_{i,r} \mu_{j,r} - \frac{1}{n-1} \sqrt{\mu_{i,r} (1 - \mu_{i,r})} \sqrt{\mu_{j,r} (1 - \mu_{j,r})}.
\]

Dividing by $\sqrt{\mu_{i,r} \mu_{j,r}} \ge \sqrt{\epsilon}$, this gives
\[
 \sqrt{\epsilon} \ge \sqrt{\mu_{i,r} \mu_{j,r}} - \frac{1}{n-1} \sqrt{(1 - \mu_{i,r}) (1-\mu_{j,r})}.
\]

Moving the second summand on the right to the left, multiplying by $n-1$, and squaring, this gives
\begin{align*}
 (n-1)^2 \mu_{i,r} \mu_{j,r} &\leq (1-\mu_{i,r})(1-\mu_{j,r}) + (n-1) \sqrt{\epsilon} \bigl(2(1-\mu_{i,r})(1-\mu_{j,r}) + (n-1)\sqrt{\epsilon}\bigr) \\ &=
 (1-\mu_{i,r})(1-\mu_{j,r}) + O(\sqrt{\epsilon}).
\end{align*}

Rearranging completes the proof.

\smallskip
\textbf{Proof of \Cref{itm:adfs-stab}.}
To prove this \namecref{itm:adfs-stab}, we first show that $\|f_{i,r}^{>1}\|, \|f_{j,r}^{>1}\|$ are small, and then appeal to~\cite[Lemma 2.4]{ADFS}.
Throughout the proof, we can assume that $\epsilon,\gamma$ are small enough, since otherwise the \namecref{itm:adfs-stab} is trivial.

As in the proof of \Cref{itm:adfs-ub}, we expand \Cref{eq:adfs}, but this time we take out one more term:
\begin{align*}
 \epsilon &\ge \mu_{i,r} \mu_{j,r} - \frac{1}{n-1} \langle f_{i,r}^{=1}, f_{j,r}^{=1} \rangle - \frac{1}{(n-1)^3} \langle f_{i,r}^{>1}, f_{j,r}^{>1} \rangle \\ &\ge
 \mu_{i,r} \mu_{j,r} - \frac{1}{n-1} \|f_{i,r}^{=1} \| \|f_{j,r}^{=1}\| - \frac{1}{(n-1)^3} \|f_{i,r}^{>1}\| \|f_{j,r}^{>1}\|,
\end{align*}
where the second inequality uses Cauchy--Schwarz (twice).

We would like to argue that $\|f_{i,r}^{=1}\| \approx \|f_{i,r}^{>0}\|$ and similarly for $j$, and to this end we define $c_{i,r},c_{j,r}$ by $\|f_{i,r}^{=1}\| = c_{i,r} \|f_{i,r}^{>0}\|$ and similarly for $j$. Since $\|f_{i,r}^{>0}\|^2 = \|f_{i,r}^{=1}\|^2 + \|f_{i,r}^{>1}\|^2$, this implies that $\|f_{i,r}^{>1}\| = \sqrt{1-c_{i,r}^2} \|f_{i,r}^{>0}\|$. Altogether, we obtain
\[
 \epsilon \ge \mu_{i,r} \mu_{j,r} - \frac{1}{n-1} \sqrt{\mu_{i,r} (1-\mu_{i,r})} \sqrt{\mu_{j,r} (1-\mu_{j,r})} \left( c_{i,r} c_{j,r} + \frac{1}{(n-1)^2} \sqrt{1 - c_{i,r}^2} \sqrt{1 - c_{j,r}^2} \right).
\]

In order to simplify this, we observe that $(1 - c_{i,r}^2) (1 - c_{j,r}^2) = (1 - c_{i,r} c_{j,r})^2 - (c_{i,r} - c_{j,r})^2$, implying that $(1 - c_{i,r})^2 (1 - c_{j,r})^2 \leq (1 - c_{i,r} c_{j,r})^2$. Therefore
\begin{align*}
 \epsilon &\ge \mu_{i,r} \mu_{j,r} - \frac{1}{n-1} \sqrt{\mu_{i,r} (1-\mu_{i,r})} \sqrt{\mu_{j,r} (1-\mu_{j,r})} \left( c_{i,r} c_{j,r} + \frac{1}{(n - 1)^2} (1 - c_{i,r} c_{j,r}) \right) \\ &=
 \mu_{i,r} \mu_{j,r} - \frac{1}{n-1} \sqrt{\mu_{i,r} (1-\mu_{i,r})} \sqrt{\mu_{j,r} (1-\mu_{j,r})} \left( \frac{1}{(n-1)^2} + \left(1 - \frac{1}{(n-1)^2}\right) c_{i,r} c_{j,r}\right).
\end{align*}

At this point we substitute the values of $\mu_{i,r}, \mu_{j,r}$, obtaining
\[
 \epsilon + O(\gamma) \ge \frac{1}{n^2} - \frac{1}{n^2} \left( \frac{1}{(n-1)^2} + \left(1 - \frac{1}{(n-1)^2}\right) c_{i,r} c_{j,r}\right) =
 \frac{1}{n^2} \left(1 - \frac{1}{(n-1)^2}\right) (1 - c_{i,r} c_{j,r}).
\]

It follows that $c_{i,r} c_{j,r} \ge 1 - O(\epsilon + \gamma)$, using $n \ge 3$. Since $c_{i,r}, c_{j,r} \leq 1$, this implies that $c_{i,r} \ge 1 - O(\epsilon + \gamma)$ and so $\sqrt{1 - c_{i,r}^2} = O(\sqrt{\epsilon + \gamma})$, and similarly for $j$. Therefore
\[
 \|f_{i,r}^{>1}\|^2, \|f_{j,r}^{>1}\|^2 = O(\epsilon + \gamma).
\]

At this point we apply \cite[Lemma 2.4]{ADFS}, which states that if $g\colon [n]^p \to \{0,1\}$ satisfies $\Ex[g] = \mu$ and $\|g^{>1}\|^2 = \eta$ then there exists a function $h\colon [n]^p \to \{0,1\}$, depending on a single coordinate, such that
\[
 \Pr[g \neq h] = O\left(\frac{\eta}{\mu - \mu^2 - \eta}\right).
\]

For small enough $\epsilon,\gamma$, this implies that there exist functions $h_{i,r},h_{j,r}$, each depending on a single coordinate, such that
\[
 \Pr[f_{i,r} \neq h_{i,r}] = O(\epsilon + \gamma),
\]
and similarly for $j$.

Since $1/n - \gamma \leq \mu_{i,r} \leq 1/n + O(\sqrt{\epsilon})$, for small enough $\epsilon,\gamma$ we see that $h_{i,r}$ must be of the form
\[
 h_{i,r}(x) = \ind{x_{k_{i,r}} = s_{i,r}},
\]
and similarly for $j$.

Observe that $h_{i,r},h_{j,r}$ are $O(\epsilon + \gamma)$-cross-independent. For small enough $\epsilon,\gamma$, this implies that $k_{i,r} = k_{j,r}$ and $s_{i,r} \neq s_{j,r}$.
\end{proof}

The next step is to determine the possible values of $\mu_{1,r},\dots,\mu_{n,r}$ for each fixed $r$.

\begin{lemma}
\label{lem:mu-*-r}
For each $r \in [n]$, one of the following holds:
\begin{enumerate}[(i)]
\item For some $i_r$ we have $\mu_{i_r,r} \ge 1 - \sqrt[4]{\epsilon}$ and $\mu_{i,r} = O(\sqrt[4]{\epsilon})$ for all $i \neq i_r$.
\label{itm:mu-*-r-const}
\item For all $i$ we have
\[
 \mu_{i,r} = \frac{1}{n} \pm O(\sqrt[4]{\epsilon}).
\]

Consequently, there exists an index $k_r \in [p]$ and a permutation $\pi_r$ on $[n]$ such that for each $i \in [n]$,
\label{itm:mu-*-r-dict}
\[
 \Pr_{x \sim [n]^p}[f_{i,r}(x) \neq \ind{x_{k_r} = \pi_r(i)}] = O(\sqrt[4]{\epsilon}).
\]
\end{enumerate}
\end{lemma}
\begin{proof}
If $\mu_{i_r,r} \ge 1 - \sqrt[4]{\epsilon}$ holds for some $i_r$ then by \Cref{lem:i-sum} we have $\mu_{i,r} = O(\sqrt[4]{\epsilon})$ for all $i \neq i_r$, which puts us in \Cref{itm:mu-*-r-const}.
In the rest of the proof, we assume that $\mu_{i,r} \leq 1 - \sqrt[4]{\epsilon}$ for all $i$, and show that \Cref{itm:mu-*-r-dict} holds.

Fix $i$. Summing \Cref{itm:adfs-ub} of \Cref{lem:adfs-approx} for all $j \neq i$ and using \Cref{lem:i-sum}, we obtain
\[
 n(n-2) \mu_{i,r} (1-\mu_{i,r} \pm O(\epsilon)) \leq (n-1)(1-\mu_{i,r}) - (1-\mu_{i,r} \pm O(\epsilon)) + O(\sqrt{\epsilon}),
\]
which simplifies to
\[
  (1-\mu_{i,r}) \cdot \mu_{i,r} \leq (1-\mu_{i,r}) \cdot \frac{1}{n} + O(\sqrt{\epsilon}).
\]

Since $1 - \mu_{i,r} \ge \sqrt[4]{\epsilon}$, this implies that
\[
 \mu_{i,r} \leq \frac{1}{n} + O(\sqrt[4]{\epsilon}).
\]

The same inequality holds for all $j \neq i$. Therefore \Cref{lem:i-sum} implies that
\[
 1 - O(\epsilon) \leq \mu_{i,r} + \sum_{j \neq i} \mu_{j,r} \leq \mu_{i,r} + \frac{n-1}{n} + O(\sqrt[4]{\epsilon}),
\]
and so
\[
 \frac{1}{n} - O(\sqrt[4]{\epsilon}) \leq \mu_{j,r}.
\]

The remaining claim follows directly from \Cref{itm:adfs-stab} of \Cref{lem:adfs-approx}.
\end{proof}

At this point we complete the proof much as in \Cref{sec:Perm-poly}. If the second case happens for some $r$ then each value of $f_i$ is attained with probability at least $1/n - O(\sqrt[4]{\epsilon})$. This conflicts with the first case, for small enough $\epsilon$. Therefore either the first case holds for all $r$, or the second case holds for all $r$.

If the first case holds for all $r$ then \Cref{itm:filmus-Perm-const} of the \namecref{thm:filmus-Perm} holds. From now on, we assume that the second case holds for all $r$.

We start by showing that all $k_r$ are equal. Indeed, if $k_r \neq k_s$ then $x_{k_r} = \pi_r(1)$ and $x_{k_s} = \pi_s(1)$ with probability $1/n^2$, and so we obtain a contradiction for small enough $\epsilon$ by considering $f_{1,r},f_{1,s}$. Thus all $k_r$ are equal to some $k$. Similarly, considering $f_i$, we see that $\pi_r(i) \neq \pi_s(i)$ for all $r \neq s$, and so there must exist a permutation $\rho_i$ such that $\pi_r(i) = \rho_i(r)$ for all $r$. Thus
\[
 \Pr_{x \sim [n]^p}[f_i(x) \neq \rho_i^{-1}(x_k)] = O(\sqrt[4]{\epsilon}).
\]
It follows that
\[
 \Pr_{y \sim \Perm_n}[(\rho_1^{-1}(y_1),\dots,\rho_n^{-1}(y_n)) \notin \Perm_n] = O(\sqrt[4]{\epsilon}).
\]

If $\rho_i \neq \rho_j$, say $\rho_i(r) \neq \rho_j(r)$, then $(y_i,y_j) = (\rho_i(r), \rho_j(r))$ occurs with probability $1/(n(n-1))$. However, in that case $\rho_i^{-1}(y_i) = \rho_j^{-1}(y_j)$, and we reach a contradiction for small enough $\epsilon$.
Thus all $\rho_i$ are equal, and so \Cref{itm:filmus-Perm-dict} holds.

\section{Equivalence relations}
\label{sec:equivalence}

\subsection{Exact polymorphisms}
\label{sec:equivalence-exact}

In this section we prove \Cref{thm:equiv-polymorphisms}. We rely on the following \namecref{lem:equiv3-polymorphisms}.

\begin{lemma}
\label{lem:equiv3-polymorphisms}
The forward direction of \Cref{thm:equiv-polymorphisms} holds when $m = 3$.
\end{lemma}
\begin{proof}
The predicate $\Equiv_3 \subseteq \{0,1\}^3$ consists of all triples whose Hamming weight is not exactly~$2$.

\cite[Theorem 1.6(5)]{Filmus26}, applied with $m = 3$ and $w = 1$, states that the polymorphisms of $\Equiv_3$ come in two types:
\begin{enumerate}[(i)]
\item All $f_{\{i,j\}}$ are equal, and the common value is $f(x) = \bigwedge_{i \in S} x_i$ for some $S \subseteq [p]$.

This corresponds to the complete equivalence relation on $[3]$.
\item At least two of the $f_{\{i,j\}}$ are constant zero.

If $f_{\{I,J\}}$ is the remaining function, then this corresponds to the equivalence relation in which the only relation is $I \sim J$. \qedhere
\end{enumerate}
\end{proof}

\begin{proof}[Proof of \Cref{thm:equiv-polymorphisms}, forward direction]
Let $(f_{\{i,j\}})_{i < j}$ be a polymorphism of $\Equiv_m$.

Define a relation on $[m]$ as follows: $i \sim i$ for all $i$, and for $i < j$, let $i \sim j$ if $f_{\{i,j\}}$ is not constant zero.

We claim that this is an equivalence relation. Indeed, suppose that $f_{\{i,j\}}$ and $f_{\{j,k\}}$ are both not constant zero. Applying \Cref{lem:equiv3-polymorphisms} to $f_{\{i,j\}},f_{\{j,k\}},f_{\{i,k\}}$, we see that $f_{\{i,j\}} = f_{\{j,k\}} = f_{\{i,k\}}$, and so $f_{\{i,k\}}$ is also not constant zero.

The first item in the \namecref{thm:equiv-polymorphisms} holds by construction. For the second item, let $A$ be an equivalence class of $\sim$ of size at least $3$. For any $i,j,k$, applying \Cref{lem:equiv3-polymorphisms} as before implies that for some set $S \subseteq [p]$,
\[
 f_{\{i,j\}}(x) = f_{\{j,k\}}(x) = f_{\{i,k\}}(x) = \bigwedge_{s \in S} x_s.
\]

To complete the proof, we show that $f_{\{\ell,r\}} = f_{\{i,j\}}$ for all $\ell,r \in A$. This is clear if $\{\ell,r\} = \{i,j\}$. If $|\{\ell,r\} \cap \{i,j\}| = 1$, say $\ell = i$, then this follows by applying \Cref{lem:equiv3-polymorphisms} to $f_{\{i,j\}},f_{\{i,r\}},f_{\{j,r\}}$. Finally, if $\{i,j\}$ and $\{\ell,r\}$ are disjoint then this follows by applying \Cref{lem:equiv3-polymorphisms} to $f_{\{i,\ell\}},f_{\{i,r\}},f_{\{\ell,r\}}$.
\end{proof}

\begin{proof}[Proof of \Cref{thm:equiv-polymorphisms}, backward direction]
Consider an arbitrary equivalence relation $\sim$ on $[m]$, and functions $(f_{\{i,j\}})_{i < j}$ satisfying the properties in the statement of the \namecref{thm:equiv-polymorphisms}.

Let $x^{(1)},\dots,x^{(p)} \in \Equiv_m$, and let $x_{\{i,j\}} = f_{\{i,j\}}(x^{(1)}_{\{i,j\}},\dots,x^{(p)}_{\{i,j\}})$. We need to show that $x \in \Equiv_m$.

Suppose that $x_{\{i,j\}} = x_{\{j,k\}} = 1$, where $i,j,k$ are all different. We need to show that $x_{\{i,k\}} = 1$.

The assumption implies that $i \sim j$ and $j \sim k$, and so $i,j,k$ belong to the same equivalence class of $\sim$. Consequently, there exists a set $S \subseteq [p]$ such that
\[
 f_{\{i,j\}}(y) = f_{\{j,k\}}(y) = f_{\{i,k\}}(y) = \bigwedge_{s \in S} y_s.
\]

Since $x_{\{i,j\}} = x_{\{j,k\}} = 1$, we have $x^{(s)}_{\{i,j\}} = x^{(s)}_{\{j,k\}} = 1$ for all $s \in S$. Therefore $x^{(s)}_{\{i,k\}} = 1$ for all $s \in S$, and so $x_{\{i,k\}} = 1$.
\end{proof}

\subsection{Approximate polymorphisms}
\label{sec:equivalence-approximate}

In this section we deduce \Cref{thm:main-equiv} by combining \Cref{thm:alekseev-filmus} and \Cref{thm:equiv-polymorphisms}.

\begin{proof}[Proof of \Cref{thm:main-equiv}]
Let $\epsilon_0$ be the parameter given by \Cref{thm:alekseev-filmus} for $\Equiv_m$ and $\mu$.
Given $\epsilon > 0$ such that $\epsilon < \epsilon_0$, we let $\delta > 0$ be the parameter given by \Cref{thm:alekseev-filmus}. Since $\mu$ is invariant under permutations of $[m]$, we have $\mu|_{\{i,j\}} = \nu$ for all $i < j$.

Suppose that $(f_{\{i,j\}})_{i < j} \colon \{0,1\}^p \to \{0,1\}$ are a $\delta$-approximate polymorphism of $\Equiv_m$ with respect to $\mu$, and moreover $\Pr_\nu[f_{\{i,j\}} = 0] < 1 - \epsilon$ for all $i < j$.

\Cref{thm:alekseev-filmus} implies that there is a polymorphism $g_{\{i,j\}}\colon \{0,1\}^p \to \{0,1\}$ of $\Equiv_m$ such that $\Pr_\nu[f_{\{i,j\}} \neq g_{\{i,j\}}] \leq \epsilon$ for all $i < j$.

\Cref{thm:equiv-polymorphisms} describes the possible forms of $g_{\{i,j\}}$, parametrized by an equivalence relation on $[m]$. We claim that this equivalence relation must be complete. Indeed, if $i\nsim j$ then $g_{\{i,j\}}$ is constant zero, and so $\Pr_\nu[f_{\{i,j\}} = 0] \ge 1 - \epsilon$, contradicting the assumption.

It follows that there exists $S \subseteq [p]$ such that
\[
 g_{\{i,j\}}(x) = \bigwedge_{s \in S} x_s
\]
for all $i < j$, from which the \namecref{thm:main-equiv} immediately follows.
\end{proof}

\section{Discussion}
\label{sec:discussion}

Some impossibility theorems in the literature concern a situation in which the individual votes satisfy one predicate, and the final outcome should satisfy a different predicate, potentially over a different alphabet. Examples include:
\begin{itemize}
\item Gibbard~\cite{GibbardTR} considered the situation in which each individual specifies a weak order on $[m]$, and the desired outcome is a partial order on $[m]$. These correspond to the following predicates (in reverse order):
\begin{itemize}
\item $\Part_m \subseteq \{0,1,*\}^{\binom{m}{2}}$ is the set of partial orders, encoded as specifying, for any $i<j$, whether $i \succ j$, $i \prec j$, or neither. We require transitivity:
\[
 i \succ j \text{ and } j \succ k \text{ implies } i \succ k.
\]
\item $\Weak_m \subseteq \{0,1,*\}^{\binom{m}{2}}$ is the set of weak orders, which are partial orders specifying the additional axiom
\[
 i \succ j \text{ implies } i \succ k \text{ or } k \succ j \text{ for each } k \neq i,j.
\]
\end{itemize}
Gibbard shows that the only unanimous polymorphisms in this setting are weak oligarchies: there exists a set $S \subseteq [p]$ such that
\begin{enumerate}[(a)]
\item If $a \stackrel k\succ b$ for all $k \in S$ then $a \succ b$.
\item If $a \succ b$ then no $k \in S$ satisfies $b \stackrel k\succ a$.
\end{enumerate}
Here $\stackrel1\succ,\dots,\stackrel p\succ$ are the individual weak orders, and $\succ$ is the aggregated partial order.
\item Dokow and Holzman~\cite{DokowHolzman10abs} considered the situation in which each individual specifies a linear order on $[m]$, and the desired outcome is a partial order on $[m]$. We encode linear orders as:
\begin{itemize}
\item $\Lin_m \subseteq \{0,1\}^{\binom{m}{2}}$ is the intersection $\Part_m \cap \{0,1\}^{\binom{m}{2}}$.
\end{itemize}
They showed that the only unanimous polymorphisms in this setting are oligarchies: the outcome is always a linear order, and there exists a set $S \subseteq [p]$ such that $a \succ b$ iff $a \stackrel k \succ b$ for all $k \in S$.
\end{itemize}

This extension of the concept of polymorphism appears in the literature on promise CSPs~\cite{AGH17,BG21,BBKO21}. In this setting, we have two predicates $P,Q$, possibly on different alphabets, such that $P$ is more restrictive than $Q$, in the sense that there is a coordinate-by-coordinate mapping from $P$ to $Q$. A $(P,Q)$ polymorphism is defined as in \Cref{def:polymorphism}, using $P$ on the left and $Q$ on the right.

\Cref{thm:alekseev-filmus} applies in this setting as well, where the flexibility condition is only needed for $P$ (and is required even if $P$ is a predicate on $\{0, 1\}$). However, we don't have an analog of \Cref{thm:filmus}. Indeed, given the examples above, it seems that we should aim for oligarchies rather than dictators.

\medskip

A different shortcoming of our framework is the unspecified dependence between $\epsilon$ and $\delta$ in \Cref{thm:alekseev-filmus}. In contrast to Mossel's quantitative Arrow theorem~\cite{Mossel12}, where the dependence is linear, the dependence in \Cref{thm:alekseev-filmus} is of tower type. We conjecture that the dependence can be improved to polynomial or even linear.

\bibliographystyle{alphaurl}
\bibliography{biblio}

\end{document}